\documentclass[english,journal]{IEEEtran}
\usepackage[T1]{fontenc}
\usepackage[latin9]{inputenc}
\usepackage{color}
\usepackage{babel}
\usepackage{amsmath}
\usepackage{amsthm}
\usepackage{amssymb}
\usepackage{graphicx}
\usepackage[unicode=true,pdfusetitle,
 bookmarks=true,bookmarksnumbered=false,bookmarksopen=false,
 breaklinks=false,pdfborder={0 0 0},pdfborderstyle={},backref=false,colorlinks=true]
 {hyperref}

\makeatletter
\theoremstyle{plain}

\theoremstyle{remark}

\theoremstyle{definition}

\theoremstyle{assumption}

\usepackage[ruled]{algorithm2e}
\newcommand{\nosemic}{\renewcommand{\@endalgocfline}{\relax}}
\newcommand{\dosemic}{\renewcommand{\@endalgocfline}{\algocf@endline}}
\usepackage{cite}
\usepackage[paper=letterpaper,top=0.75in,bottom=0.75in,right=0.75in,left=0.75in]{geometry}
\usepackage{subfigure}
\usepackage{enumerate}
\newtheorem{assumption}{Assumption}

\makeatother

\newtheorem{remark}{Remark}

\newtheorem{theorem}{Theorem}

\providecommand{\definitionname}{Definition}
\providecommand{\remarkname}{Remark}
\providecommand{\theoremname}{Theorem}

\begin{document}
\newgeometry{top=1in,bottom=0.75in,right=0.75in,left=0.75in}
\IEEEoverridecommandlockouts 
\title{Parameter Update Laws for Adaptive Control with Affine Equality Parameter Constraints}
\author{Ashwin P. Dani  \thanks{A. P. Dani is with the Department of Electrical and Computer Engineering at University of Connecticut, Storrs, CT 06269.}}
\maketitle

\begin{abstract}
In this paper, constrained parameter update laws for adaptive control with convex equality constraint on the parameters are developed, one based on a gradient only update and the other incorporating concurrent learning (CL) update. The update laws are derived by solving a constrained optimization problem with affine equality constraints. This constrained problem is reformulated as an equivalent unconstrained problem in a new variable, thereby eliminating the equality constraints. The resulting update law is integrated with an adaptive trajectory tracking controller, enabling online learning of the unknown system parameters. Lyapunov stability of the closed-loop system with the equality-constrained parameter update law is established. The effectiveness of the proposed equality-constrained adaptive control law is demonstrated through simulations, validating its ability to maintain constraints on the parameter estimates, achieving convergence to the true parameters for CL-based update law, and achieving asymptotic and exponential tracking performance for constrained gradient and constrained CL-based update laws, respectively.
\end{abstract}

\section{Introduction}
Adaptive control is a control method which achieves control objectives in the presence of parametric model uncertainties. It has a rich history of developments over several decades \cite{annaswamy2021historical,ioannou1996robust,krstic1995nonlinear}. Problems such as parameter drifts and the identification of model parameters while achieving controller tracking/regulation objectives have been addressed in several papers. In this paper, we address the problem of injecting prior knowledge about the parameters in the form of constraints on the parameter update laws. Specifically, we design parameter update laws that incorporate equality constraints on the parameter estimates, given in the form of $A\theta = d$. 

Various parameter update laws have been proposed for adaptive control. The gradient update law is one of the most basic forms studied. To prevent parameter drift, $\sigma$ and $e$-modification algorithms are developed. Smooth projection algorithms keep the parameters bounded within a prescribed convex set \cite{cai2006sufficiently,ioannou1996robust}. The $\mathcal{L}_1$-adaptive controller uses a low-pass filter in the feedback loop to enforce desired transient performance while allowing fast adaptation. Parameter update laws that use both tracking and prediction errors, referred to as composite adaptation \cite{Slotine1991} have been developed, and can achieve fast parameter convergence while achieving the control objective. These parameter update laws typically assume persistency of excitation (PE) of the regressor to achieve parameter convergence. Using least squares technique of regression problem, parameter update laws called concurrent learning (CL) and integral concurrent learning (ICL) are developed. The CL/ICL-based parameter update laws can achieve parameter convergence with finite excitation (FE) condition, which is a weaker/milder condition than PE. Dynamic regressor extension and mixing (DREM) can achieve parameter convergence under milder conditions than PE \cite{ortega2020new}. Using the techniques in optimization for parameter learning, momentum-based tools have been used to design parameter update laws in \cite{somers2024online,le2024accelerated}, with the goal of improving transient performance. For deep neural network, parameter update laws are designed in \cite{le2021real}.

In certain applications, additional constraints on the parameters, such as symmetry with eigenvalue bounds or positive definiteness of parameter matrices, are important for ensuring physically consistent parameter estimates and improving controller robustness. The methods in \cite{boffi2021implicit, lee2018natural} develop parameter update laws based on Bregman divergence metrics to enforce positive definiteness of parameter matrices. The method in \cite{moghe2022projection} incorporates symmetric matrix constraints with eigenvalue bounds into the parameter estimation law. These approaches enable the incorporation of specific structural constraints present in the system dynamics. In \cite{hart2023lyapunov}, an adaptive controller that incorporates knowledge of dynamics is developed using Lyapunov-based control and parameter update law design. In some applications, it is critical to strictly enforce upper and lower bound constraints on the parameters. For example, this arises when estimating control effectiveness matrix $g(x)$ in the control-affine system $\dot{x}=f(x)+g(x)u$ \cite{DaniAACRL2025}, or when parameter constraints require maintaining the norm of the parameter vector within prescribed bounds. Constrained parameter update laws using Barrier methods, as studied in optimization, are developed in \cite{DaniCDC2025,DaniAACRL2025}.

In this paper, a new parameter update law that incorporates affine equality constraints of the form $A\theta = d$ on the parameter estimates is developed. The optimization viewpoint of developing parameter update law shown in \cite{boffi2021implicit,fradkov1979speed,ioannou1996robust} is used to incorporate such constraints. As shown in \cite{boffi2021implicit,fradkov1979speed}, the objective function used to derive gradient update law is convex. When affine equality constraints on the parameters (which are convex constraints) are added, the resulting parameter estimation problem becomes a convex program. Two new equality-constrained parameter update laws are designed - one based on a gradient update law and the other based on a CL-based update law. The objective function is convex for the gradient-based update law, and strictly convex for the CL-based update law. To design parameter update laws with equality constraints, an equivalent unconstrained optimization is formulated by projecting the parameter estimates onto the constrained set \cite{Boyd2004}. This yields an  optimization problem in a new reduced-dimensional variable $z(t)$, where the parameter estimate $\hat{\theta}(t)$ is a function of $z(t)$. Explicit formulas for $\hat{\theta}(t)$ in terms of $z(t)$ are derived, where $z(t)$ is updated using a differential equation, resulting in a reduced-dimensional parameter update law. Stability of the parameter update law, together with the tracking error dynamics induced by a tracking controller, is established via Lyapunov analysis. Using this constraint-elimination formulation of the convex program, there is no need to introduce Lagrange multipliers to enforce the parameter constraints, which simplifies the stability analysis and yields a reduced-dimensional update law. Both equality-constrained gradient and CL-based parameter update laws are tested in simulations using constraints of the form $\theta_i-\theta_j=0$. The simulations validate that the proposed update laws maintain the parameter constraints for all time while achieving the control objectives. In the simulations, the parameters may not converge to true values for the equality-constrained gradient-based update law, whereas the parameters convergence is obtained for equality-constrained CL-based update law.

Notations: For any matrix $A\in \mathbb{R}^{n \times m}$, $\mathcal{N}(A)$ denotes the null space, $\mathcal{R}(A)$ denotes the range space. For a square matrix $A$, $A\succ 0$ denotes positive definite matrix $A$

\section{System Model and Control Objective}
\subsection{System Dynamics}
Consider the following system
\begin{equation}
\dot{x}=f(x) + u \label{eq:SystemModel}
\end{equation}
where $x(t) \in \mathbb{R}^n$ is the system state, $u\in\mathbb{R}^n$ is the control input, $f:\mathbb{R}^n\rightarrow \mathbb{R}^n$ is a locally Lipschitz continuous function which is linearly parametrizable, defined as
\begin{align}
f(x) = Y(x)\theta \label{eq:LinParam}
\end{align}
where $Y: \mathbb{R}^n \rightarrow \mathbb{R}^{n \times p}$ is the regressor matrix and $\theta \in \mathbb{R}^p$ is an unknown parameter. The parameters are such that they satisfy following affine equality constraints
\begin{equation}
    A\theta = d \label{eq:EqualityConstraint}
\end{equation}
where $d \in \mathbb{R}^m$ and $A \in \mathbb{R}^{m\times p}$ are constants.
An example of such a constraint would be
\begin{equation}
    \exists i,j: \theta_i(t) - \theta_j(t) = 0, \quad i,j \in \{1,..,p\}, \;\;i \neq j
\end{equation}
\begin{remark}
    The equality constraints are convex since the constraints are affine in $\theta$. \label{rem:rem1}
\end{remark}
\begin{remark}
    Such parameter constraints imply that there is some prior knowledge about parameters, e.g., robot link masses are equal to each other for robot dynamics. \label{rem:rem2}
\end{remark} 
\begin{assumption}
    The matrix $A$ is full row rank and $\exists$ $\kappa_1, \kappa_2>0$ such that $\kappa_1 I \leq AA^T \leq \kappa_2I$. \label{ass:Assumption1}
\end{assumption}

\subsection{Controller Objective}
The control objective is to track a desired trajectory $x_d(t) \in \mathbb{R}^n$ and compute parameter vector estimate $\hat{\theta}(t)\in \mathbb{R}^p$ while maintaining prescribed equality constraints of the form (\ref{eq:EqualityConstraint}) on the parameter estimates. For the control design, let's define tracking and parameter estimation errors as follows
\begin{equation}
    e(t) = x(t) - x_d(t), \quad
    \tilde{\theta}(t) = \theta - \hat{\theta}(t)
    \label{eq:regulation_error}
\end{equation}
where $\tilde{{\theta}}(t)\in\mathbb{R}^p$ is the parameter estimation error.

\subsection{Adaptive Controller Design} \label{sec:ControlDesign}
To achieve the control objective, an adaptive trajectory tracking controller is designed as
\begin{equation}
u = \dot{x}_d - Y\hat{\theta} - ke \label{eq:Control}
\end{equation}
where $\dot{x}_d(t)$ is the derivative of the desired trajectory, and $k>0$ is the control gain. Taking the time derivative of (\ref{eq:regulation_error}) and substituting the controller (\ref{eq:Control}), the following closed-loop error system can be obtained 
\begin{equation}
    \dot{e} = Y\tilde{\theta} -ke  \label{eq:ErrorDyn}
\end{equation}
The parameter estimation error dynamics is written as
\begin{equation}
\dot{\tilde{\theta}} = -\dot{\hat{\theta}}
\end{equation}
In next sections, gradient and CL-based parameter update laws are presented that satisfy the affine equality constraints of the form (\ref{eq:EqualityConstraint}).

\section{Equality-Constrained Gradient Update Law}
In this section, an equality-constrained gradient parameter update law is designed. The affine equality constraint on the parameter estimates can be written using (\ref{eq:EqualityConstraint}) as follows
\begin{equation}
    A\hat{\theta} = d. \label{eq:EqualityConstraintEst}
\end{equation}
To design the constrained parameter update law, consider the following constrained minimization problem
\begin{align}
    &\hat{\theta}^* = \mathrm{min}_{\hat{\theta}} \; \gamma e^TY\tilde{\theta}  \nonumber \\
    \mathrm{s.\;t.}&  \; c(\hat{\theta}) = A\hat{\theta} - d = 0,  \label{eq:ConGradProb1}
\end{align}
where $\gamma > 0$ is a constant gain. An equivalent constraint free formulation of the convex optimization can be written by formulating a new variable $z \in \mathbb{R}^{p-m}$ (see \cite{Boyd2004} Chapter 4, Eliminating Equality Constraints) 
\begin{align}
    \hat{\theta}_0 &= A^T(AA^T)^{-1}d \nonumber \\
    \hat{\theta} &= \hat{\theta}_0 + Fz \label{eq:newVarGrad}
\end{align}
where $\hat{\theta}_0 \in \mathbb{R}^p$ represents a solution to equality constraint $c(\hat{\theta})=0$, $F \in \mathbb{R}^{p\times(p-m)}$ is selected such that $\mathcal{R}(F) = \mathcal{N}(A)$. 
\begin{remark}
    One way to choose F is to perform singular value decomposition (SVD) of $A=U \Sigma V^T$ and choose last $p-m$ columns of $V$ corresponding to $\mathcal{N}(A)$. This will ensure that $F$ is full column rank, i.e., rank $p-m$ and $\mathcal{R}(F) = \mathcal{N}(A)$. \label{rem:rem3}
\end{remark}
The optimization problem can be reformulated in new variable $z$ and the equality constraint can be eliminated as follows
\begin{align}
    z^* = \mathrm{min}_{z} f(z) = \mathrm{min}_{z} \; \gamma e^TY(\theta - \hat{\theta}_0 - Fz)   \label{eq:UnconGradProbNewVar}
\end{align}
The convex optimization problems (\ref{eq:ConGradProb1}) and (\ref{eq:UnconGradProbNewVar}) are equivalent (see \cite{Boyd2004}) and the optimal solution satisfies
\begin{equation}
    \hat{\theta}^* =  \hat{\theta}_0 + Fz^* \label{eq:optimalThetaGrad}
\end{equation}
The gradient dynamics of $z$ can be computed using (\ref{eq:UnconGradProbNewVar}) as
\begin{equation}
    \dot{z} = -\nabla_z f(z) = \gamma F^TY^Te 
\end{equation}
The parameter update law in reduced dimension is given by
\begin{align}
    \hat{\theta} &= \hat{\theta}_0 + Fz \nonumber \\
    \dot{z} &= \gamma F^TY^Te  \label{eq:EqualityConstGradUpdateLaw}
\end{align}
\subsection{Error dynamics}
The tracking error dynamics from (\ref{eq:ErrorDyn}) is expressed as
\begin{align}
    \dot{e} = -k e + Y\tilde{\theta}
\end{align}
Using (\ref{eq:EqualityConstGradUpdateLaw}) the parameter estimation error dynamics is given by
\begin{align}
    \dot{\tilde{\theta}} &= -\dot{\hat{\theta}} = -\dot{\hat{\theta}}_0 - F\dot{z} \nonumber \\
    \dot{\tilde{\theta}} &= - \gamma F F^TY^Te     \label{eq:thetaTildeGrad} 
\end{align}

\subsection{Stability Analysis}
In this subsection, stability of the tracking error dynamics (\ref{eq:ErrorDyn}) and equality-constrained gradient parameter update law (\ref{eq:EqualityConstGradUpdateLaw}) is established.
\begin{theorem}
    If Assumption \ref{ass:Assumption1} is satisfied, for the system shown in (\ref{eq:SystemModel}), the equality-constrained parameter update law (\ref{eq:EqualityConstGradUpdateLaw}) and the adaptive controller (\ref{eq:Control}) ensure global asymptotic tracking, i.e., 
    \begin{equation}
        \Vert e(t) \Vert \rightarrow 0 \quad \mathrm{as} \quad t \rightarrow \infty
    \end{equation} and bounded parameter estimation error with constrained satisfaction on the parameter estimates.
\end{theorem}
\begin{proof}
    Let $y = [e^T, \tilde{\theta}^T]^T \in \mathbb{R}^{n+p}$ be an auxiliary vector. Consider a Lyapunov function
\begin{equation}
    V(y) = \frac{1}{2}e^Te + \frac{1}{2\gamma}\tilde{\theta}^T\tilde{\theta}
\end{equation}
The bounds on $V(y)$ are
\begin{equation}
    \lambda_1\Vert y \Vert^2 \leq V(y) \leq \lambda_2 \Vert y \Vert^2
\end{equation}
where $\lambda_1 = \mathrm{min}\left(\frac{1}{2},\frac{1}{2\gamma}\right)$, $\lambda_2 = \mathrm{max}\left(\frac{1}{2},\frac{1}{2\gamma}\right)$. Time derivative of $V(y)$ yields
\begin{align}
    \dot{V} &= e^T(-ke +Y\tilde{\theta}) + \tilde{\theta}^T(-FF^TY^Te) \nonumber \\
    \dot{V} &= e^T(-ke) + \tilde{\theta}^T(Y^Te-FF^TY^Te) \nonumber \\
    \dot{V} &= -ke^Te + \tilde{\theta}^T(I-FF^T)Y^Te
\end{align}
Since $\theta$ satisfies $A\theta = d$ and $\hat{\theta}$ is designed to satisfy $A\hat{\theta} = d$, we have
\begin{equation}
    A(\theta - \hat{\theta}) = 0, \implies A\tilde{\theta} = 0
\end{equation}
This implies that $\tilde{\theta} \in \mathcal{N}(A)$. Since $F$ is constructed using the orthonormal columns of the $\mathcal{N}(A)$ (Remark \ref{rem:rem3}), $FF^T$ is an orthonormal projector onto  $\mathcal{N}(A)$. So, for any $\tilde{\theta} \in \mathcal{N}(A)$
\begin{equation}
    FF^T\tilde{\theta} = \tilde{\theta} \implies (I - FF^T)\tilde{\theta} = 0 \label{eq:orthoNormProj}
\end{equation}
Using (\ref{eq:orthoNormProj}) the derivative of Lyapunov function can be written as
\begin{equation}
    \dot{V} = -ke^Te \leq 0
\end{equation}
Since $V(y)$ is positive definite (PD) and $\dot{V}$ is negative semi-definite (NSD), the system (\ref{eq:ErrorDyn})-(\ref{eq:thetaTildeGrad}) is stable in the sense of Lyapunov, which means $y(t)\in\mathcal{L}_\infty$. Computing the derivative of $\dot{V}$  yields
\begin{equation}
    \ddot{V} = -ke^T(-ke+Y\tilde{\theta})
\end{equation}
Since $e(t), \tilde{\theta}(t) \in \mathcal{L}_\infty$, $\ddot{V} \in \mathcal{L}_\infty$, which means $\dot{V}$ is uniformly continuous \cite{Slotine1991}. Using Barbalat's lemma $\dot{V} \rightarrow 0$ as $t \rightarrow \infty$, which implies $\Vert e \Vert \rightarrow 0$ as $t \rightarrow \infty$.
\end{proof}

\section{Equality-Constrained Concurrent Learning Update Law}
In this section, CL-based parameter update law with affine equality constraints is derived. For the CL-based equality-constrained parameter update law, prior data of state and control input called as a history stack $\mathcal{H} = \{x_k,\hat{u}_k,\dot{\hat{x}}_k\}_{k=1}^{k=h}$ is collected. Before discussing the parameter update law, finite excitation (FE) assumption is stated.
\begin{assumption} 
For the history stack $\mathcal{H} = \{x_k,\hat{u}_k,\dot{\hat{x}}_k\}_{k=1}^{k=m}$ the following condition is satisfied
\begin{equation}
    \sigma_2I \geq Y_R \geq \sigma_1I
\end{equation}
where $Y_R = \sum_{k=1}^h Y_k^T Y_k$, $\sigma_1, \sigma_2 \in \mathbb{R}^+$. The numerically computed derivatives of $x(t)$, $\dot{\hat{x}}_k$ computed at $k$th data point satisfies $\Vert \dot{\hat{x}}_k - \dot{x}_k \Vert \leq \epsilon$ for a small positive $\epsilon \in \mathbb{R}^+$. \label{ass:FiniteExcitation}
\end{assumption}
\begin{remark}
Assumption \ref{ass:FiniteExcitation} is a finite excitation condition that can be verified in real-time \cite{chowdhary2013concurrent}.
\end{remark}
To design the equality-constrained CL parameter update law, consider a convex constrained minimization problem
\begin{align}
    &\hat{\theta}^* = \mathrm{min}_{\hat{\theta}} \; \gamma e^TY\tilde{\theta} + \frac{k_{cl}}{2} \tilde{\theta}^TY_R\tilde{\theta}  \nonumber \\
    \mathrm{s.\;t.}&  \; c(\hat{\theta}) = A\hat{\theta} - d = 0,  \label{eq:ConProb1}
\end{align}
where $k_{cl}$ is a positive constant. Similar to previous section, an equivalent equality constraint free formulation of the convex optimization can be written by formulating a new variable
\begin{align}
    \hat{\theta}_0 &= A^T(AA^T)^{-1}d \nonumber \\
    \hat{\theta} &= \hat{\theta}_0 + Fz \label{eq:newVar}
\end{align}
where $\hat{\theta}_0 \in \mathbb{R}^p$, $F \in \mathbb{R}^{p\times(p-m)}$ and $z \in \mathbb{R}^{p-m}$ are defined in (\ref{eq:newVarGrad}). 

The optimization problem can be written in new variable $z$ as an unconstrained convex optimization problem
\begin{align}
    z^* =& \mathrm{min}_{z} f_{CL}(z) \nonumber\\
    z^* =& \mathrm{min}_{z} \; \gamma e^TY(\theta - \hat{\theta}_0 - Fz) \nonumber \\
    &+ \frac{k_{cl}}{2} (\theta - \hat{\theta}_0 - Fz)^TY_R(\theta - \hat{\theta}_0 - Fz)  \label{eq:UnconProbNewVar}
\end{align}
The convex optimization problems (\ref{eq:ConProb1}) and (\ref{eq:UnconProbNewVar}) are equivalent (see Section 4.1.3 of \cite{Boyd2004}), which means the optimal solution satisfies
\begin{equation}
    \hat{\theta}^* =  \hat{\theta}_0 + Fz^* \label{eq:optimalTheta}
\end{equation}
The gradient dynamics of $z$ can be computed by
\begin{equation}
    \dot{z} = -\nabla_z f_{CL}(z) = \gamma F^TY^Te + k_{cl}F^TY_R\tilde{\theta}
\end{equation}
The parameter update law is given by
\begin{align}
    \hat{\theta} &= \hat{\theta}_0 + Fz \nonumber \\
    \dot{z} &= \gamma F^TY^Te + k_{cl}F^TY_R\tilde{\theta} \label{eq:ConsCLUnImpl}
\end{align}
Since $\tilde{\theta}(t)$ is not measurable, the update law in (\ref{eq:ConsCLUnImpl}) is not implementable, but can be used for analysis. In an implementable form the equality-constrained CL-based parameter update law is given by
\begin{align}
    \hat{\theta} &= \hat{\theta}_0 + Fz \nonumber \\
    \dot{z} &= \gamma F^TY^Te + k_{cl}F^T\sum_{k=1}^hY_k^T(\dot{\hat{x}}_k - u_k - Y_k \hat{\theta}) \label{eq:ConsCLImpl}
\end{align}
where (\ref{eq:SystemModel}) and (\ref{eq:LinParam}) are used.
\subsection{Error dynamics}
The tracking error dynamics using (\ref{eq:ErrorDyn}) can be expressed as
\begin{equation}
    \dot{e} = -k e + Y\tilde{\theta}
\end{equation}
Using (\ref{eq:ConsCLUnImpl}) the parameter estimation error dynamics is given by
\begin{equation}
    \dot{\tilde{\theta}} = -\dot{\hat{\theta}} =  -\gamma FF^TY^Te - k_{cl}FF^TY_R\tilde{\theta} \label{eq:thetaTildeCL}
\end{equation}

\subsection{Stability analysis}
In this subsection, stability of the tracking error dynamics (\ref{eq:ErrorDyn}) and equality-constrained CL-based parameter update law (\ref{eq:ConsCLImpl}) is established using Lyapunov analysis.

\begin{theorem}
    If Assumptions \ref{ass:Assumption1}-\ref{ass:FiniteExcitation} are satisfied, for the system shown in (\ref{eq:SystemModel}), the equality-constrained parameter update law (\ref{eq:ConsCLImpl}) and the adaptive controller (\ref{eq:Control}) ensure globally exponentially stable tracking performance with constrained satisfaction on the parameter estimates and bounded closed-loop signals.
\end{theorem}
\begin{proof}
Let $y = [e^T, \tilde{\theta}^T]^T \in \mathbb{R}^{n+p}$ be an auxiliary vector. Consider a Lyapunov function
\begin{equation}
    V(y) = \frac{1}{2}e^Te + \frac{1}{2\gamma}\tilde{\theta}^T\tilde{\theta} \label{eq:LyapFunThm2}
\end{equation}
The bounds on $V(y)$ are
\begin{equation}
    \lambda_1\Vert y \Vert^2 \leq V(y) \leq \lambda_2 \Vert y \Vert^2 \label{eq:LyapBoundsThm2}
\end{equation}
where $\lambda_1 = \mathrm{min}\left(\frac{1}{2},\frac{1}{2\gamma}\right)$, $\lambda_2 = \mathrm{max}\left(\frac{1}{2},\frac{1}{2\gamma}\right)$. Using (\ref{eq:ErrorDyn}) and (\ref{eq:thetaTildeCL}), the time derivative of $V(y)$ is computed as
\begin{equation}
    \dot{V} = e^T(-ke +Y\tilde{\theta}) + \tilde{\theta}^T(-FF^TY^Te - \frac{k_{cl}}{\gamma}FF^TY_R \tilde{\theta}) 
\end{equation}
Using $FF^T\tilde{\theta} = \tilde{\theta}$ from (\ref{eq:orthoNormProj}) yields
\begin{align}
    \dot{V} &= -ke^Te - \frac{k_{cl}}{\gamma}\tilde{\theta}^TY_R \tilde{\theta} 
\end{align}
Using Assumption \ref{ass:FiniteExcitation} $Y_R \succ 0$, which leads to
\begin{equation}
    \dot{V} = -ke^Te - \frac{k_{cl}\sigma_1}{\gamma}\tilde{\theta}^T\tilde{\theta} = -\mathrm{min}\left( 2k, 2k_{cl}\sigma_1\right)V. \label{eq:VDotThm2}
\end{equation}
The following bound can be developed on $y$ using (\ref{eq:VDotThm2}) and (\ref{eq:LyapBoundsThm2})
\begin{equation}
    \Vert y(t) \Vert \leq \mu_o\Vert y(t_0) \Vert e^{-\mu_1 t}
\end{equation}
where $\mu_0 = \sqrt{\frac{\lambda_2}{\lambda_1}}$ and $\mu_1 = \mathrm{min}\left( 2k, 2k_{cl}\sigma_1\right)$.
Since $\Vert y \Vert \in \mathcal{L}_\infty$, $\Vert \tilde{\theta} \Vert \in \mathcal{L}_\infty$, which implies $\hat{\theta}\in \mathcal{L}_\infty$. Using (\ref{eq:ConsCLUnImpl}), since $\hat{\theta}_0$ and $F$ are constants, $z \in \mathcal{L}_\infty$.
\end{proof}

\begin{remark}
    For the time period when the history stack data is collected, $\lambda_{min}(\sum_{k=1}^m Y_k^T Y_k)$ is not full rank, thus, Assumption \ref{ass:FiniteExcitation} is not satisfied. For this time period, the equality-constrained parameter update law (\ref{eq:EqualityConstGradUpdateLaw}) can be used, which ensures that all the closed-loop signals are bounded.
\end{remark}

\begin{remark}
    The data points $\mathcal{H}$ collected in history stack $\sum_{k=1}^m Y_k^T Y_k$ can be replaced using singular value maximization algorithm in \cite{chowdhary2013concurrent}, which ensures $\sum_{k=1}^m Y_k^T Y_k$ is always increasing, thus, Lyapunov function (\ref{eq:LyapFunThm2}) serves as a  common Lyapunov function \cite{parikh2019integral}.
\end{remark}

\section{Simulations}
Simulations are carried out to demonstrate the applicability of the equality-constrained parameter update law in adaptive control setting. The following system dynamics similar to \cite{parikh2019integral} is considered
\begin{align}
    \dot{x} = \left[ \begin{matrix}
        x_1^2 & sin(x_2) & 0 & 0 \\
        0 & x_2sin(x_1) & x_1 & x_1x_2
    \end{matrix}\right]\theta + u
\end{align}
where $x = [x_1\;x_2]^T \in \mathbb{R}^2$ is the state, $u\in\mathbb{R}^2$ is the control input, $\theta \in \mathbb{R}^4$ are the true parameters whose values are given by $\theta = [5 \; 5 \; 10 \; 20]^T$. The true parameters satisfy equality constraint $\theta_1 = \theta_2$, which can be written as (\ref{eq:EqualityConstraint}) using 
\begin{equation}
    A = \left[\begin{matrix}
        1 & -1 & 0 & 0
    \end{matrix}\right],\;\; d = 0.
\end{equation}
The desired trajectory is given by
\begin{equation}
    x_d(t) = 10(1-e^{-0.1t})\left[ \begin{matrix}
    \mathrm{sin}(2t) \\ 0.4\mathrm{cos}(3t) 
    \end{matrix}\right]
\end{equation}

\subsection{Simulation for Equality-Constrained Gradient Law}

The controller in (\ref{eq:Control}) is implemented along with the equality-constrained parameter update law given in (\ref{eq:EqualityConstGradUpdateLaw}). The system state is initialized to $x(t_0) = [10\;5]^T$, the parameters are initialized to $\hat{\theta}(t_0) = [4.5\;4.5\;4.5\;15]^T$. This satisfies the parameter constraint at the initial time. The control gain is selected as $k=2 \mathrm{diag}\{10,50\}$ and $\gamma = 0.4$, $F$ is computed using $\mathcal{N}(A)$ given by \begin{equation}
    F = \left[\begin{smallmatrix}
    0.7071 &  0.7071 & 0 & 0 \\
    0 & 0 & 1 & 0 \\
    0 & 0 & 0 & 1
\end{smallmatrix}\right]^T
\end{equation} and $\hat{\theta}_0 = 0$.

The performance of the equality-constrained adaptive controller is shown in Figs. \ref{fig:TrajTrackingSim1} and \ref{fig:ParamComparisonSim1}. It is observed from Fig. \ref{fig:ParamComparisonSim1} that the parameters estimated by the proposed equality-constrained parameter update law satisfy the parameter constraints, i.e., $\hat{\theta}_1(t) = \hat{\theta}_2(t)$ and the tracking errors converge to zero.

\begin{figure}
    \centering
    \includegraphics[width=\linewidth]{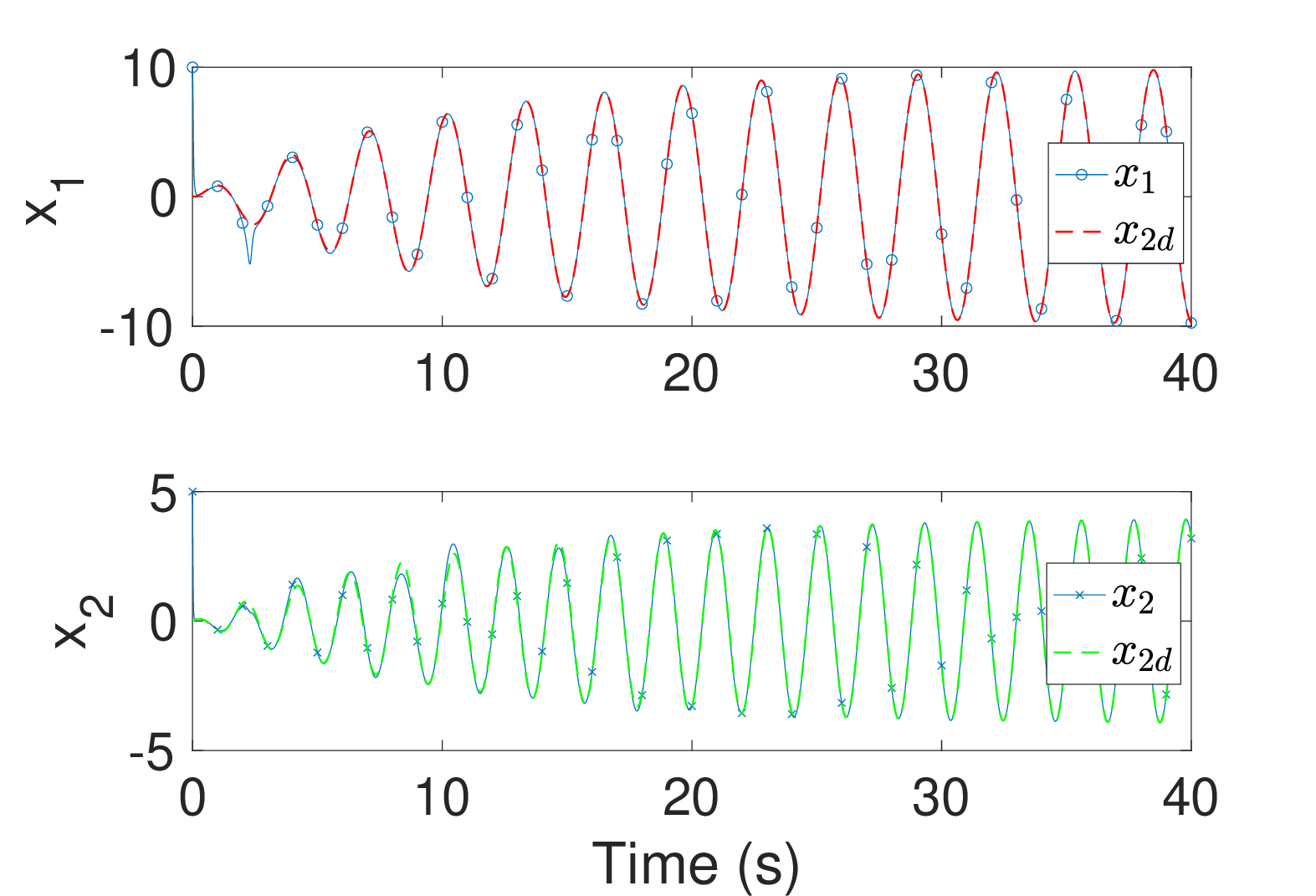}
    \caption{Trajectory tracking with equality-constrained gradient parameter update law.}
    \label{fig:TrajTrackingSim1}
\end{figure}
\begin{figure}
    \centering
    \includegraphics[width=1.0\linewidth]{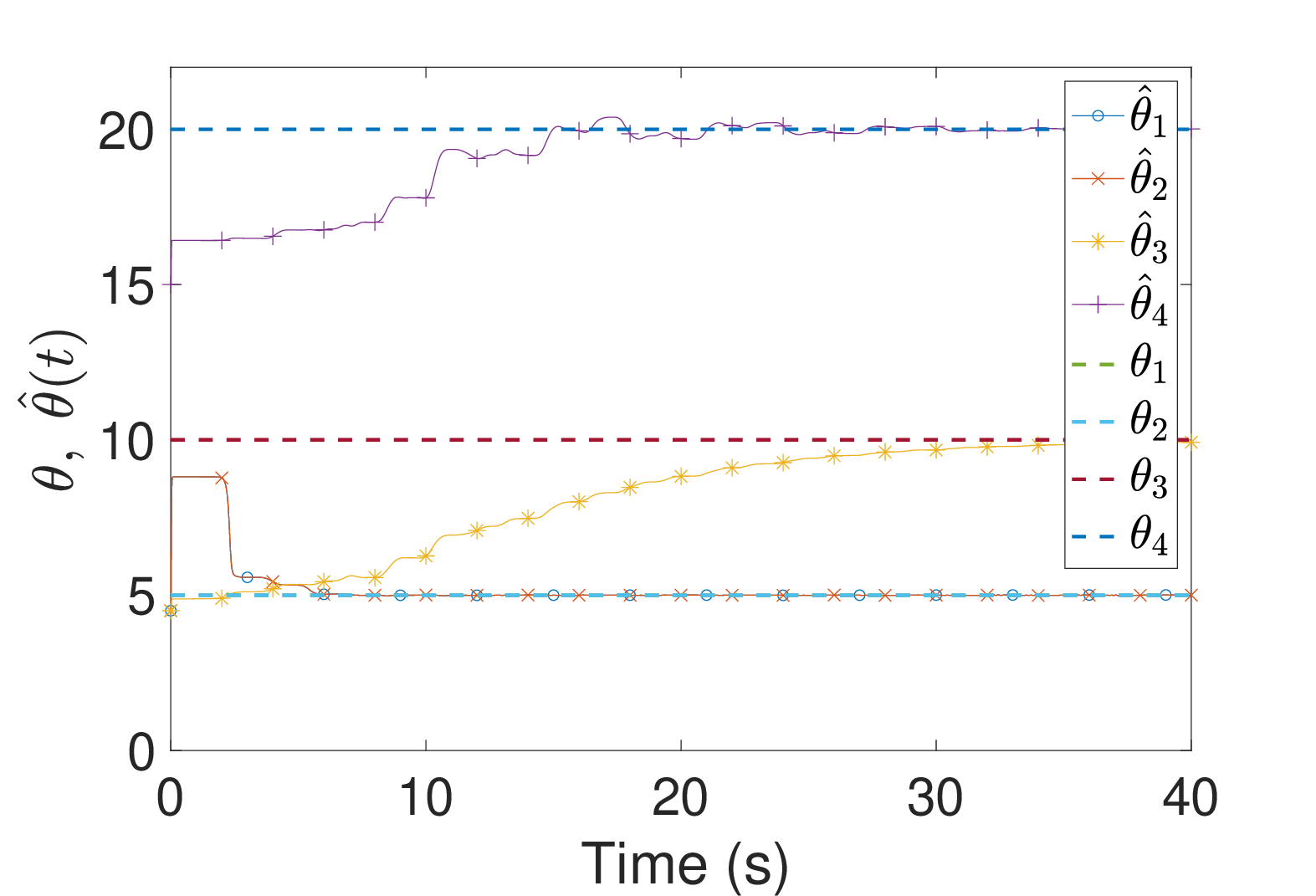}
    \caption{Parameter estimates using equality-constrained gradient parameter update law.}
    \label{fig:ParamComparisonSim1}
\end{figure}

\subsection{Simulation for Equality-Constrained CL Update Law}
In this simulation, equality-constrained CL-based update law is tested. The controller is implemented using constraints on $\hat{\theta}$ for the same dynamics and desired trajectories as that of previous simulation use case but with different parameter set $\theta = [5\;20\;10\;20]^T$. 
The parameter estimates are initialized to $\hat{\theta}(t_0) = [3\;10\;5\;10]^T$ such that the parameter constraints are satisfied. The control gain is selected as $k=2\mathrm{diag}\{10.50\}$ and parameter adaptation law gains are $\gamma=0.05$ and $k_{cl} = 0.0008$. In this case, $A = [0 \; -1 \; 0 \; 1]$ and $d=0$, for which $F$ is computed using $\mathcal{N}(A)$ given by 
\begin{equation}
    F = \left[\begin{smallmatrix}
    0.7071 &  0.5 & 0 & 0.5 \\
    0 & 0 & 1 & 0 \\
    -0.7071 & 0.5 & 0 & 0.5
\end{smallmatrix}\right]^T
\end{equation}
and $\hat{\theta}_0 = 0$.

From the simulation results shown in Figs. \ref{fig:TrajTrackingSim2}-\ref{fig:ParamComparisonECCLSim2}, it is observed that the parameters are estimated while the equality constraints on the parameters are maintained for all time, i.e., $\hat{\theta}_2(t) = \hat{\theta}_4(t)$ and the tracking error converges to zero. The parameter estimates using gradient update law are shown in Fig. \ref{fig:ParamComparisonGradSim2} using the same gains used for the equality-constrained CL update law. The parameters may not be identified correctly by the gradient update law if the regressor is not PE.

\begin{figure}
    \centering
    \includegraphics[width=\linewidth]{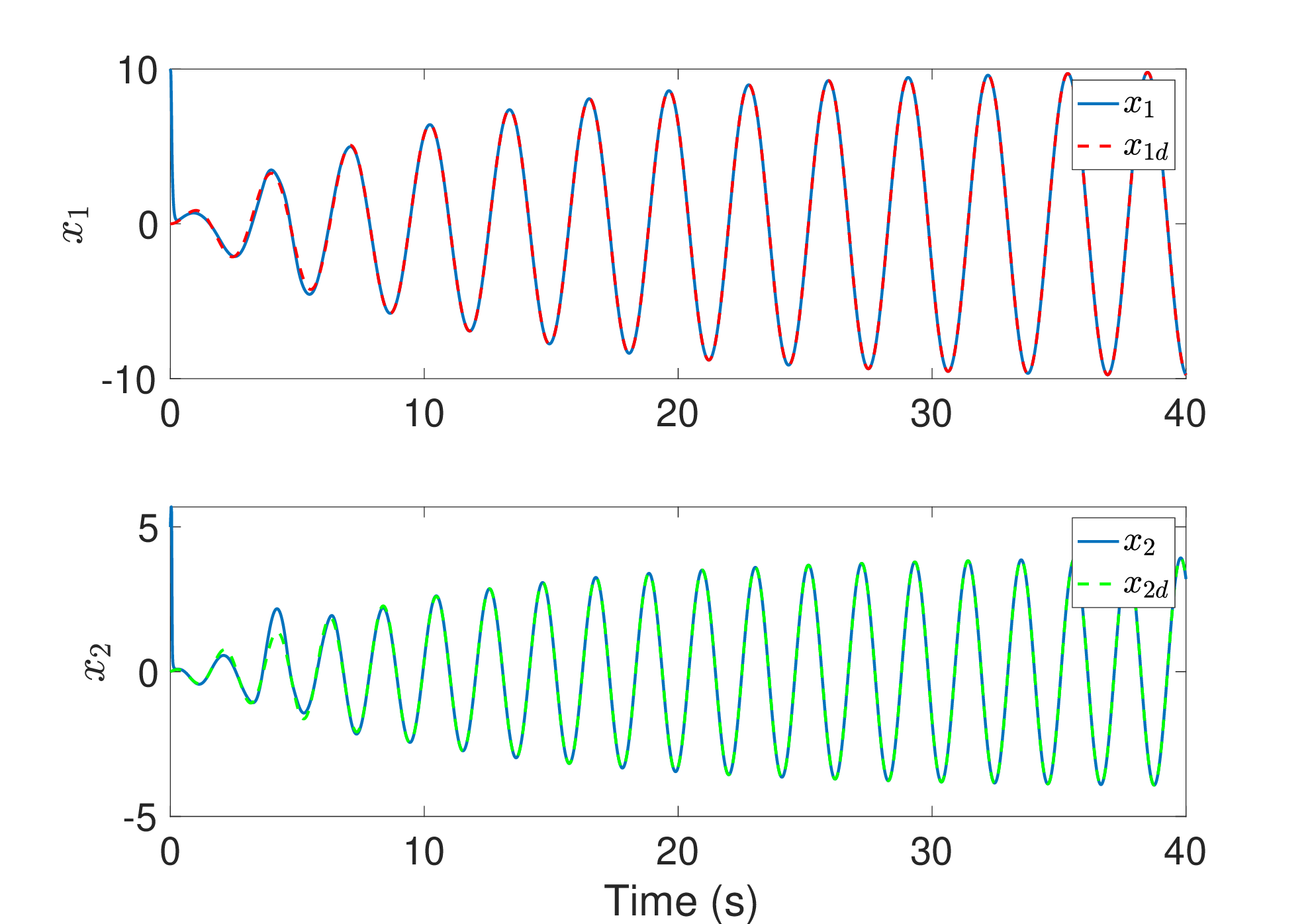}
    \caption{Trajectory tracking with equality-constrained CL parameter update law.}
    \label{fig:TrajTrackingSim2}
\end{figure}
\begin{figure}
    \centering
    \includegraphics[width=1.0\linewidth]{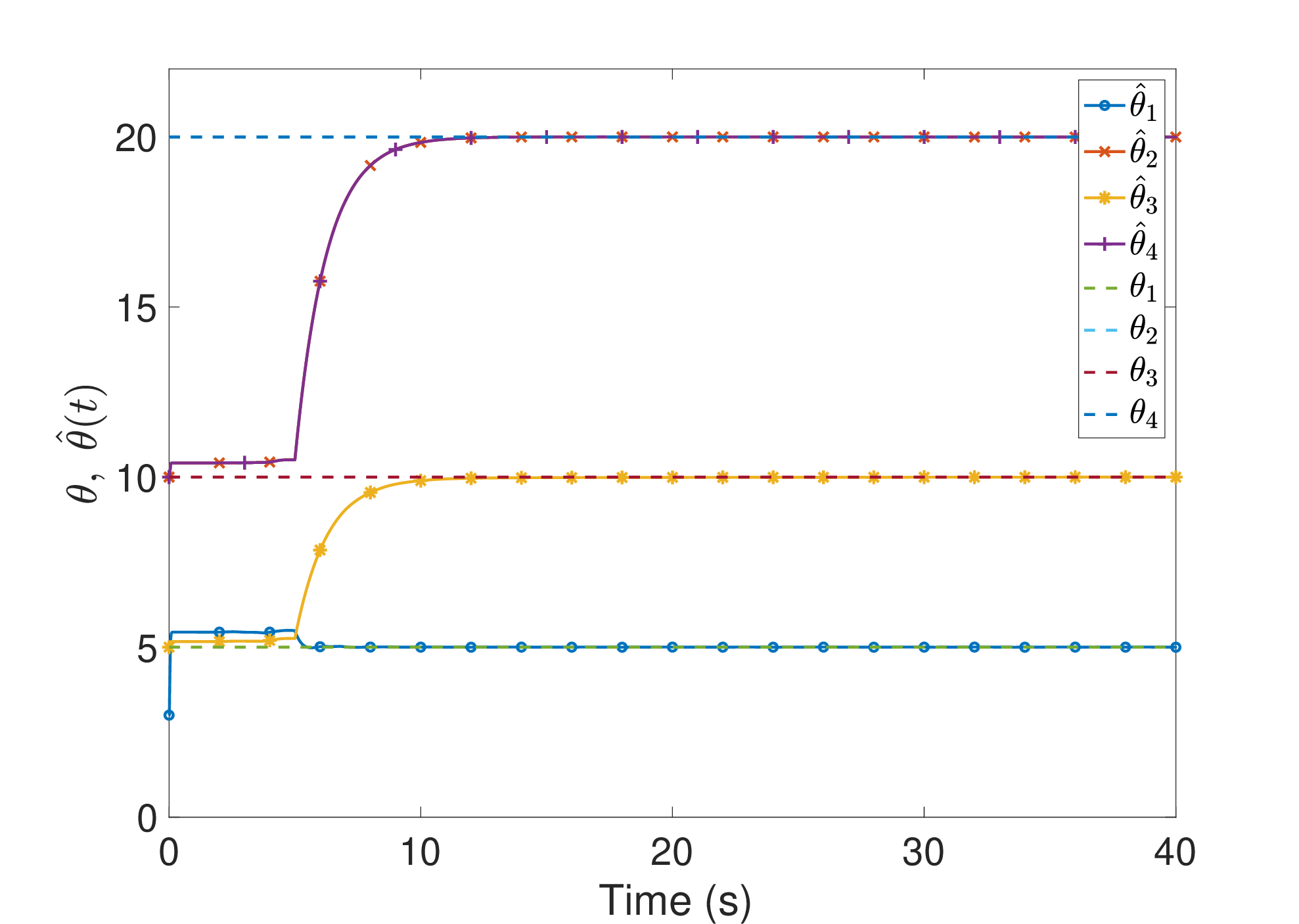}
    \caption{Parameter estimates using equality-constrained CL parameter update law, where $\hat{\theta}_2$ and $\hat{\theta}_4$ satisfy the equality constraint $\hat{\theta}_2-\hat{\theta}_4=0$.}
    \label{fig:ParamComparisonECCLSim2}
\end{figure}
\begin{figure}
    \centering
    \includegraphics[width=1.0\linewidth]{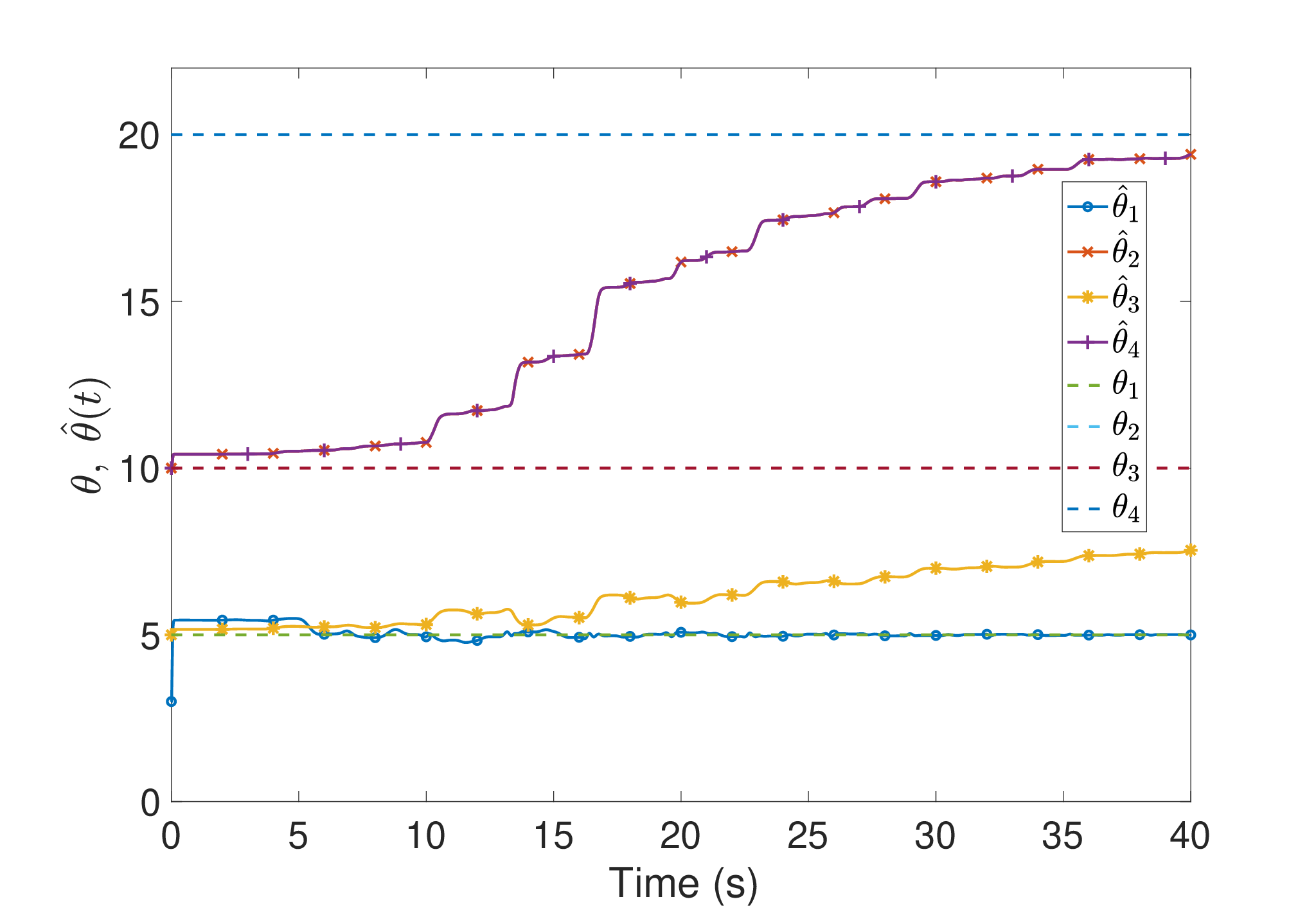}
    \caption{Parameter estimates using equality-constrained gradient update law, where $\hat{\theta}_2$ and $\hat{\theta}_4$ satisfy the equality constraint $\hat{\theta}_2-\hat{\theta}_4=0$.}
    \label{fig:ParamComparisonGradSim2}
\end{figure}

\section{Conclusion}
In this paper, an equality-constrained parameter update law is developed for adaptive tracking control using a convex optimization formulation. The approach of deriving adaptive parameter update law using a minimization problem of an objective function subject to affine equality constraints on the parameter estimates yields continuous update law solutions in a modified form for both gradient and CL-based update laws. Lyapunov stability analysis of the tracking error, parameter estimation error dynamics is conducted which shows that the tracking error converge asymptotically and parameter estimation error remains bounded for gradient-based update law and exponentially under FE condition for the CL-based update law. By construction, the affine equality constraints on the parameter estimates are always satisfied. Simulation results validate the proposed method in both the cases where equality constraints are preserved on the parameter estimates.

\bibliographystyle{IEEEtran}
\bibliography{RCL_Complete}

@book{ioannou1996robust,
  title={Robust adaptive control},
  author={Ioannou, Petros A and Sun, Jing},
  volume={1},
  year={1996},
  publisher={PTR Prentice-Hall Upper Saddle River, NJ}
}

@article{boffi2021implicit,
  title={Implicit regularization and momentum algorithms in nonlinearly parameterized adaptive control and prediction},
  author={Boffi, Nicholas M and Slotine, Jean-Jacques E},
  journal={Neural Computation},
  volume={33},
  number={3},
  pages={590--673},
  year={2021}
}

@BOOK{Boyd2004,
  title = {Convex Optimization},
  publisher = {Cambridge University Press},
  year = {2004},
  author = {Boyd, S. and Vandenberghe, L.}
}

@ARTICLE{chowdhary2013concurrent,
  author = {Chowdhary, Girish and Yucelen, Tansel and Muhlegg, Maximillian and
	Johnson, Eric N},
  title = {Concurrent learning adaptive control of linear systems with exponentially
	convergent bounds},
  journal = {International Journal of Adaptive Control and Signal Processing},
  year = {2013},
  volume = {27},
  pages = {280--301},
  number = {4},
  publisher = {Wiley Online Library}
}

@BOOK{Slotine1991,
  title = {Applied Nonlinear Control},
  publisher = {Englewood Cliffs, NJ: Prentice-Hall},
  year = {1991},
  author = {Slotine, J.-J. E. and Li, W.}
}

@article{parikh2019integral,
  title={Integral concurrent learning: Adaptive control with parameter convergence using finite excitation},
  author={Parikh, Anup and Kamalapurkar, Rushikesh and Dixon, Warren E},
  journal={International Journal of Adaptive Control and Signal Processing},
  volume={33},
  number={12},
  pages={1775--1787},
  year={2019}
}

@article{fradkov1979speed,
  title={Speed-gradient scheme and its application in adaptive control problems},
  author={Fradkov, AL},
  journal={Automation and Remote Control},
  volume={9},
  pages={90--101},
  year={1979}
}

@inproceedings{lee2018natural,
  title={A natural adaptive control law for robot manipulators},
  author={Lee, Taeyoon and Kwon, Jaewoon and Park, Frank C},
  booktitle={IEEE/RSJ International Conference on Intelligent Robots and Systems},
  pages={1--9},
  year={2018}
}

@article{moghe2022projection,
  title={Projection scheme and adaptive control for symmetric matrices with eigenvalue bounds},
  author={Moghe, Rahul and Akella, Maruthi R},
  journal={IEEE Transactions on Automatic Control},
  volume={68},
  number={3},
  pages={1738--1745},
  year={2022}
}

@ARTICLE{DaniAACRL2025,
  author={Dani, Ashwin P. and Bhasin, Shubhendu},
  journal={IEEE Open Journal of Control Systems}, 
  title={Adaptive Actor-Critic Based Optimal Regulation for Drift-free Nonlinear Systems}, 
  year={2025},
  volume={4},
  pages={117-129}
}

@inproceedings{DaniCDC2025,
    author = {Dani, Ashwin P},
    title = {Constrained Parameter Update Law using Inverse Barrier Function for Adaptive Control},
    booktitle = {IEEE Conference on Decision and Control},
    year = {2025},
    pages = {}
}

@article{somers2024online,
  title={Online accelerated data-driven learning for optimal feedback control of discrete-time partially uncertain systems},
  author={Somers, Luke and Haddad, Wassim M and Kokolakis, Nick-Marios T and Vamvoudakis, Kyriakos G},
  journal={International Journal of Adaptive Control and Signal Processing},
  volume={38},
  number={3},
  pages={848--876},
  year={2024}
}

@article{cai2006sufficiently,
  title={A sufficiently smooth projection operator},
  author={Cai, Zhijun and de Queiroz, Marcio S and Dawson, Darren M},
  journal={IEEE Transactions on Automatic Control},
  volume={51},
  number={1},
  pages={135--139},
  year={2006}
}

@article{ortega2020new,
  title={New results on parameter estimation via dynamic regressor extension and mixing: Continuous and discrete-time cases},
  author={Ortega, Romeo and Aranovskiy, Stanislav and Pyrkin, Anton A and Astolfi, Alessandro and Bobtsov, Alexey A},
  journal={IEEE Transactions on Automatic Control},
  volume={66},
  number={5},
  pages={2265--2272},
  year={2020}
}

@article{annaswamy2021historical,
  title={A historical perspective of adaptive control and learning},
  author={Annaswamy, Anuradha M and Fradkov, Alexander L},
  journal={Annual Reviews in Control},
  volume={52},
  pages={18--41},
  year={2021}
}

@book{krstic1995nonlinear,
  title={Nonlinear and adaptive control design},
  author={Krstic, Miroslav and Kokotovic, Petar V and Kanellakopoulos, Ioannis},
  year={1995},
  publisher={John Wiley \& Sons, Inc.}
}

@article{le2024accelerated,
  title={Accelerated gradient approach for deep neural network-based adaptive control of unknown nonlinear systems},
  author={Le, Duc M and Patil, Omkar Sudhir and Nino, Cristian F and Dixon, Warren E},
  journal={IEEE Transactions on Neural Networks and Learning Systems},
  volume={36},
  number={4},
  pages={6299--6313},
  year={2024}
}

@article{le2021real,
  title={Real-time modular deep neural network-based adaptive control of nonlinear systems},
  author={Le, Duc M and Greene, Max L and Makumi, Wanjiku A and Dixon, Warren E},
  journal={IEEE Control Systems Letters},
  volume={6},
  pages={476--481},
  year={2021}
}

@article{hart2023lyapunov,
  title={Lyapunov-based physics-informed long short-term memory {(LSTM)} neural network-based adaptive control},
  author={Hart, Rebecca G and Griffis, Emily J and Patil, Omkar Sudhir and Dixon, Warren E},
  journal={IEEE Control Systems Letters},
  volume={8},
  pages={13--18},
  year={2023}
}

\end{document}